\documentclass[10pt, doublecolumn]{IEEEtran}
\usepackage{epsfig,latexsym}
\usepackage{float}
\usepackage{indentfirst}
\usepackage{amsmath}
\usepackage{amssymb}
\usepackage{times}
\usepackage{subfigure}
\usepackage{psfrag}
\usepackage{hyperref}
\usepackage{cite}
\usepackage{algorithm}
\usepackage[noend]{algpseudocode}
\usepackage{lastpage}
\usepackage{fancyhdr}
\usepackage{color}
 \usepackage{amsthm}
\usepackage{bigints}
\newtheorem{theorem}{Theorem}
\newtheorem{Lemma}{Lemma}
\newtheorem{Corollary}{Corollary}
\newtheorem{lemma}[Lemma]{$\mathbf{Lemma}$}

\newtheorem{corollary}[Corollary]{$\mathbf{Corollary}$}

\begin{document}%%
\title{ \huge{      Delay Minimization for  NOMA-MEC Offloading  }}

\author{ Zhiguo Ding,  Derrick Wing Kwan Ng,    Robert Schober,  and H. Vincent Poor  \thanks{   \vspace{-1.5em}

   Z. Ding and H. V. Poor are  with the Department of
Electrical Engineering, Princeton University, Princeton,
USA.  Z. Ding  is  also  with the School of
Electrical and Electronic Engineering, the University of Manchester, Manchester,  UK. D. W. K. Ng is with the School of Electrical Engineering and
Telecommunications, University of New South Wales, Sydney, Australia.   
R. Schober is with the Institute for Digital Communications,
Friedrich-Alexander-University Erlangen-Nurnberg (FAU), Germany.
  
  }\vspace{-2.3em}} \maketitle

\begin{abstract}  
This paper considers the  minimization of the offloading delay for non-orthogonal multiple access assisted mobile edge computing (NOMA-MEC). By transforming  the delay minimization problem into a form of fractional programming,  two iterative algorithms based on Dinkelbach's method and Newton's method are proposed. The optimality  of both methods is proved and their convergence is compared. Furthermore, criteria for  choosing between   three possible modes, namely orthogonal multiple access (OMA), pure NOMA, and hybrid NOMA, for MEC offloading are   established.     \vspace{-1.7em}
\end{abstract} 
 
 \section{Introduction}
The application of non-orthogonal multiple access (NOMA) to mobile edge computing (MEC) has received considerable attention recently \cite{MECding,8269088,8267072,Zhiguo_MEC1}. In particular, the superior performance   of NOMA-MEC    with fixed    resource allocation was illustrated  in \cite{MECding}. In  \cite{8269088}, a  weighted sum-energy minimization problem was investigated in a multi-user NOMA-MEC system. In \cite{8267072}, the energy consumption of NOMA-MEC networks was minimized   assuming that each user has access to multiple bandwidth resource blocks.    In \cite{Zhiguo_MEC1}, joint power and time allocation was designed   for NOMA-MEC, again with the objective to minimize the offloading energy consumption. To the best of the authors' knowledge, the minimization of the offloading delay for NOMA-MEC has not yet been studied. 

The aim of this letter is to study   delay minimization for NOMA-MEC offloading. Compared to the energy minimization problems studied in \cite{8269088, 8267072, Zhiguo_MEC1},  minimizing the offloading delay is more challenging, since the delay is the ratio of two rate-related  functions.  We first transform the formulated delay minimization problem into a form of fractional programming. However,  the   transformed problem  is fundamentally different from the original fractional programming  problem   in \cite{Dinkelbach}. To this end, two algorithms based on Dinkelbach's method and Newton's method are proposed, and their optimality is rigorously proved. While for conventional fractional programming, the two methods are equivalent,   as for the problem considered in this paper,   Newton's method is proved to converge faster than   Dinkelbach's method. In addition, criteria for  choosing between   three possible modes, namely  OMA, pure NOMA, and hybrid NOMA (H-NOMA), for MEC offloading are   established.   Interestingly, we find that  pure NOMA can outperform H-NOMA when there is sufficient energy for MEC offloading, whereas H-NOMA always outperforms pure NOMA if the objective is to reduce the energy consumption as shown in~\cite{Zhiguo_MEC1}.

 \vspace{-0.3em}
\section{System Model} 
Consider an   MEC offloading  scenario, in which two  users, denoted by user $m$ and user $n$,   offload their computation tasks to an MEC server. Without loss of generality, assume that the two users' tasks contain the same number of nats, denoted by $N$, and user $m$'s computation deadline, denoted by $D_m$ seconds, is shorter    than user $n$'s, denoted by $D_n$ seconds, i.e., $D_m\leq D_n$.  
 
In this work, as user $m$ has a more stringent delay requirement than user $n$, user $n$ will be served in an   opportunistic manner  as described in the following. In particular, user $m$'s  transmit power, denoted by $P_m$, is set the same as in  OMA, i.e., $P_m$ satisfies  $
  D_m\ln(1+P_m|h_m|^2)=N$, 
where $h_{k}$ denotes user $k$'s channel gain, $k\in\{m,n\}$.  User $n$ is allowed to access the time slot of $D_m$ seconds allocated to user $m$, under the condition that user $m$ experiences the same rate as for   OMA.  As shown in \cite{Zhiguo_MEC1}, this can be realized if  user $m$'s  message is decoded after user $n$'s   at the MEC server  and    user $n$'s data rate, denoted by $R_n$, during $D_m$ is set  as follows: 
\begin{align}\label{upper bound rn}
R_n= \ln \left(1+\frac{P_{n,1}|h_n|^2}{P_m |h_m|^2+1}\right),
\end{align}
where $P_{n,1}$ denotes the power used by user $n$ during $D_m$. If user $n$ cannot finish its offloading within $D_m$, a dedicated time slot, denoted by $T_n$, is allocated to user $n$, and   the user's transmit power during $T_n$ is denoted by $P_{n,2}$.  Note that the three cases with $\{P_{n,1}=0, P_{n,2}\neq0\}$, $\{P_{n,1}\neq0, P_{n,2}=0\}$,    and $\{P_{n,1}\neq0, P_{n,2}\neq0\}$, correspond to OMA, pure NOMA, and H-NOMA, respectively. We note that both OMA and pure NOMA can be viewed as special cases of H-NOMA. However,   in this paper, the three modes are considered separately  and H-NOMA is restricted to the case with $\{P_{n,1}\neq0, P_{n,2}\neq0\}$. Similar to \cite{8269088, 8267072, Zhiguo_MEC1}, the time and energy costs for the users to download the computation outcomes from the MEC server are omitted, as they are negligibly small compared to the considered   uploading costs.

\section{ NOMA-Assisted MEC Offloading}
The considered opportunistic strategy can guarantee that user $m$'s delay performance in NOMA is the same as that in OMA, and   the   problem for minimizing user $n$' delay  can be formulated as follows: 
\begin{subequations}\label{5}
\begin{eqnarray}\label{5a}&
\underset{ P_{n,1},P_{n,2}}{\text{ minimize}}\quad& D_m+T_n\\ \label{5b}
&\text{s.t.} \quad& \left\{P_{n,1}  P_{n,2}\right\}\in\mathcal{S},
\end{eqnarray}
\end{subequations}
where      $T_n=\frac{N- D_m\ln\left(1+\frac{P_{n,1}|h_n|^2}{P_m |h_m|^2+1}\right)}{\ln\left(1+|h_n|^2P_{n,2}\right) }$, 
{\small \begin{align}\label{s1}
\mathcal{S}=\{ D_mP_{n,1}+T_nP_{n,2}\leq E, P_{n,1}\geq 0,P_{n,2}\geq 0, T_n\geq0\},\end{align}}
\hspace{-1em} and $E$ is user $n$'s   energy   constraint. 
 We note that $T_n$ is zero if pure NOMA is used, and   the constraint in Eq. \eqref{5b} implies that   user $n$'s power during $D_m$ needs to ensure that user $m$ experiences the same rate as in OMA.  Define $ E_1=D_m\left(e^{\frac{N}{D_m}}-1\right)  |h_n|^{-2}$ and {\small $E_2= D_m\left(e^{\frac{N}{D_m}}-1\right)e^{\frac{N}{D_m}} |h_n|^{-2}$}.  The optimal    power allocation policy  depends on how much energy is available as shown in the following subsections. 
\vspace{-1em}
\subsection{Case $E\geq E_2$}
This  corresponds to the case with   sufficient energy at user $n$ for MEC offloading,  and the minimal $D_m$ can be achieved by using pure NOMA, i.e., $P_{n,2}=0$. The condition for adopting  pure NOMA  is shown  in the following.   Since  $P_{n,2}=0$, all the energy is consumed during the NOMA phase to minimize the delay, which means $P_{n,1}=\frac{E}{D_m}$. Hence, user $n$ is able to offload its task within $D_m$ if
\begin{align}\nonumber
N\leq &D_m\ln\left(1+\frac{P_{n,1}|h_n|^2}{P_m|h_m|^2+1}\right)
\\  \overset{(a)}{=}& D_m\ln\left(1+ e^{-\frac{N}{D_m}}\frac{E}{D_m}|h_n|^2 \right) , \label{constraint 1}
\end{align}
where step (a) is obtained by assuming that user $m$'s power, $P_m$, satisfies the constraint $D_m\ln\left(P_m|h_m|^2+1\right)=N$.  By solving the inequality in \eqref{constraint 1}, the condition  $E\geq E_2$ can be obtained. 

The performance gain of NOMA-MEC over OMA-MEC  is obvious   in this case since user $n$'s delay in OMA  is $D_m+\frac{N }{\ln\left(1+|h_n|^2P_{n,2}\right) }$ which is strictly larger   than   $D_m$. However,   the comparison between OMA and NOMA becomes more complicated for the energy-constrained cases. 
 
 \vspace{-1em}
 \subsection{Case $ E_1<E< E_2$}\label{subsectionxxx}

This corresponds to the case, where there is not sufficient energy at user $n$ to support pure NOMA. Note that both  H-NOMA and OMA are still applicable. Due to the space limit, we focus on  H-NOMA  ($P_{n,i}\neq0$, $i\in\{1,2\}$), as the OMA solution can be obtained in a straightforward manner.  

Note that $T_n$ is the ratio  of two functions of $P_{n,1}$ and $P_{n,2}$, respectively, which motivates the use of fractional programming. However, compared to conventional  fractional programming in \cite{Dinkelbach}, the problem  in \eqref{5} is more challenging since the fractional function $T_n$ does not only appear in the objective function but also in the constraint. Two iterative algorithms will be developed based on  the following  Dinkelbach's auxiliary function parameterized by $\mu$:
\begin{align}\label{8}
F(\mu) =& \underset{ P_{n,1},P_{n,2}}{\text{maxmize}}\quad \ln\left(1+|h_n|^2P_{n,2}\right) \\\nonumber & -  \mu \left(N- D_m\ln\left(1+ e^{-\frac{N}{D_m}}P_{n,1}|h_n|^2 \right)\right),
\end{align} 
where $\left\{P_{n,1}, P_{n,2}\right\}\in\tilde{\mathcal{S}}(\mu)$ and\footnote{We note that for the case   $E< E_2$, $T_n$ is always non-negative as shown in \eqref{constraint 1}, and hence the constraint, $T_n\geq 0$, can be omitted. } 
{\small \begin{align}\label{s2}
\tilde{\mathcal{S}}(\mu)=\left\{D_mP_{n,1}+\mu^{-1}P_{n,2} \leq E, P_{n,1}\geq 0,P_{n,2}\geq 0\right\}.
\end{align}}
Different from  the original form   in \cite{Dinkelbach}, 
  the constraint set $\tilde{\mathcal{S}}(\mu)$ for the auxiliary function  is also a function of $\mu$ and $F(\mu)$ might have more than one root.

For a fixed $\mu$, the following lemma provides the optimal H-NOMA    solution for problem \eqref{8}.  
\begin{lemma}
For a fixed $\mu$, the optimal H-NOMA power allocation policy for problem \eqref{8} is given by
\begin{eqnarray}\label{solution hnoma}
\left\{\begin{array}{rl}  P_{n,1}^{*}(\mu) &=\frac{E-\mu^{-1}\left(e^{\frac{N}{D_m}}-1\right)|h_n|^{-2}}{D_m+\mu^{-1}}\\
P_{n,2}^{*}(\mu) &= \frac{E+D_m\left(e^{\frac{N}{D_m}}-1\right)|h_n|^{-2}}{D_m+\mu^{-1}}
     \end{array}\right. .
\end{eqnarray} 
\end{lemma}
\begin{proof}
Due to the space limit, only  a  sketch of the proof is provided. 
It is straightforward to show that problem \eqref{8} is convex for a fixed $\mu$. Hence the Karush-Kuhn-Tucker (KKT) can be applied to find the optimal solution as follows: \cite{Boyd}
 {\small \begin{eqnarray}\nonumber
\left\{\begin{array}{rl} -\frac{\mu D_m e^{-\frac{N}{D_m}} |h_n|^2}{1+ e^{-\frac{N}{D_m}}P_{n,1}|h_n|^2}+\lambda_1D_m -\lambda_2-\lambda_3 &=0\\-\frac{|h_n|^2}{1+|h_n|^2P_{n,2}} +\lambda_1\mu^{-1} -\lambda_2-\lambda_3&=0 \\  D_mP_{n,1}+\mu^{-1}P_{n,2}&\leq E \\ 
     \lambda_1\left(D_mP_{n,1}+\mu^{-1}P_{n,2}- E \right)&=0 \\
     P_{n,i}&\geq 0, \forall i \in\{1,2\}\\
     \lambda_iP_{n,i}&=0, \forall i \in\{2,3\}\\
     \lambda_i&\geq 0, \forall i \in\{1, 2,3\}
     \end{array}\right.\hspace{-2em},
\end{eqnarray}}
\hspace{-0.5em}where $\lambda_i$ are Lagrange multipliers. 
For   the H-NOMA case,  $P_{n,1}> 0 $ and $P_{n,2}> 0$, and hence  $\lambda_2=0$ and $\lambda_3=0$ due to   constraints $ \lambda_iP_{n,i}=0, \forall i \in\{2,3\}$. Therefore, the KKT conditions can be simplified as follows:
\begin{eqnarray}
\left\{\begin{array}{rl} -\frac{\mu D_m e^{-\frac{N}{D_m}} |h_n|^2}{1+ e^{-\frac{N}{D_m}}P_{n,1}|h_n|^2}+\lambda_1D_m  &=0\\-\frac{|h_n|^2}{1+|h_n|^2P_{n,2}} +\lambda_1\mu^{-1}  &=0 \\  D_mP_{n,1}+\mu^{-1}P_{n,2}&\leq E \\ 
     \lambda_1\left(D_mP_{n,1}+\mu^{-1}P_{n,2}- E \right)&=0 \\
     P_{n,i}&\geq 0, \forall i \in\{1,2\}\\
     \lambda_1&\geq 0
     \end{array}\right.\hspace{-1em}.
\end{eqnarray}
Due to constraint $-\frac{|h_n|^2}{1+|h_n|^2P_{n,2}} +\lambda_1\mu^{-1}  =0 $,  we have $\lambda_1\neq 0$ for the optimal solution, otherwise   $\frac{|h_n|^2}{1+|h_n|^2P_{n,2}}=0$ which cannot be true. Since $\lambda_1\neq 0$ and $ \lambda_1\left(D_mP_{n,1}+\mu^{-1}P_{n,2}- E \right)=0$, we have $  D_mP_{n,1}+\mu^{-1}P_{n,2}- E  =0$. With some algebraic manipulations, the lemma can be proved. 
\end{proof}

{\it Remark 1:} By substituting  $P_{n,1}^{*}(\mu)$ and $P_{n,2}^{*}(\mu)$ into \eqref{8}, $F(\mu)$ can be expressed as an explicit function of $\mu$. Unlike \cite{Dinkelbach},   $F(\mu)$ is not convex and may have multiple roots, which result in fundamental changes to the proof of   convergence of   Dinkelbach's method. 

Although $F(\mu)$ is different from the form in \cite{Dinkelbach},  a modified Dinkelbach's method can still be   developed as shown in Algorithm \ref{algorithm}.   In Algorithm \ref{algorithm}, $\delta$ denotes a small positive    threshold. In addition, a   Newton's method-based iterative   algorithm can also be developed   as in Algorithm \ref{algorithm2}, where $F'(\mu)$ denotes the first order derivative of $F(\mu)$.

 \begin{algorithm}
\caption{ Dinkelbach's Method Based Algorithm} 
\begin{algorithmic}[1]
 
\State Set $t=0$, $\mu_0=+\infty$  , $\delta\rightarrow 0$.  
\While { $F(\mu_{t})<-\delta$  }
\State $t=t+1$. 
\State  Update $P_{n,1}^t$ and $P_{n,2}^t$ by using     $\mu_{t-1}$ and  \eqref{solution hnoma}. 
\State Update    $F(\mu_t)$ by using $P_{n,1}^t$ and $P_{n,2}^t$. 
\State Update $\mu_t$ as $\mu_t= \frac{\ln\left(1+|h_n|^2P^t_{n,2}\right)}{   N- D_m\ln\left(1+ e^{-\frac{N}{D_m}}P_{n,1}^t|h_n|^2 \right)} $

\EndWhile
\State \textbf{end}
\State  $P_{n,1}^{*}=P_{n,1}^t$ and $P_{n,2}^{*}=P_{n,2}^t$. 
\end{algorithmic}\label{algorithm}
\end{algorithm}

 \begin{algorithm}
\caption{Newton's Method Based Algorithm} 
\begin{algorithmic}[1]
 
\State Set $t=0$,  $\mu_0=+\infty$, $\delta\rightarrow 0$.
\While { $F(\mu_{t})<-\delta$  }
\State $t=t+1$.  
\State Update $\mu_{t+1}=\mu_t -\frac{F(\mu_t)}{F'(\mu_t)}$. 

\EndWhile
\State \textbf{end}
\State  $P_{n,1}^{*,N}$ and $P_{n,2}^{*,N}$ are obtained by using $\mu_t$ and \eqref{solution hnoma}. 
\end{algorithmic}\label{algorithm2}
\end{algorithm}\vspace{-1em}

The following theorem confirms  the optimality of the two developed algorithms.  

\begin{theorem}\label{theorem1}
The two iterative algorithms shown in Algorithms \ref{algorithm} and \ref{algorithm2} converge to the same optimal solution of problem \eqref{5}. 
\end{theorem} 

\vspace{-1em}
\begin{proof}
We start the proof by first studying the roots of the function in \eqref{8} which can potentially be  the optimal solution.  We note that  the steps provided in \cite{Dinkelbach} cannot be straightforwardly applied to prove the optimality of the modified Dinkelbach's method  since $F(\mu)$ may have multiple roots.  

{\it \underline{Step 1}} of the proof is to show the existence and uniqueness of the root of $F(\mu)$, for $ \left(e^{\frac{N}{D_m}}-1\right)|h_n|^{-2} E^{-1}<\mu<\infty$.  To prove this, it is sufficient to show that 1): $F(\mu)$ is strictly concave, 2): $F(\mu)=-\infty$ for $\mu\rightarrow +\infty$ and 3): $F(\mu)>0$ for $\mu\rightarrow \left(e^{\frac{N}{D_m}}-1\right)|h_n|^{-2} E^{-1}$.

The first order derivative of $F(\mu)$ is given by
\begin{align}\label{first derival}
 F'(\mu)  = &\left.\frac{|h_n|^2E+D_m\left(e^{\frac{N}{D_m}}-1\right)}{e^{\frac{N}{D_m}}D_m\mu +|h_n|^2E\mu+1}-\right[N - D_m  \\\nonumber &\left. \times \ln\left(1+ e^{-\frac{N}{D_m}}|h_n|^2 \frac{E-\frac{\left(e^{\frac{N}{D_m}}-1\right)|h_n|^{-2}}{\mu}}{D_m+\frac{1}{\mu}}\right)\right].
\end{align}
Unlike \cite{Dinkelbach}, $ F'(\mu)$ is neither strictly negative nor positive. The second order derivative of $F(\mu)$ is given by
\begin{align}
F''(\mu)  = &\frac{\left(D_m^2 - \left(e^{\frac{N}{D_m}}D_m +|h_n|^2E\right)^2\right)}{\left(e^{\frac{N}{D_m}}D_m\mu +|h_n|^2E\mu+1\right)^2(\mu D_m+1)} .
\end{align}
Since $e^{\frac{N}{D_m}}>1$ and $|h_n|^2E>0$, we have
\begin{align} \label{ea1}
 F''(\mu)  <0,
\end{align}
which means that $F(\mu)$ is a strictly concave function of $\mu$.

When $\mu\rightarrow \left(e^{\frac{N}{D_m}}-1\right)|h_n|^{-2} E^{-1}$, the H-NOMA solution approaches  $P_{n,1}^*(\mu)\rightarrow 0$ and $P_{n,2}^*(\mu)\rightarrow \left(e^{\frac{N}{D_m}}-1\right)|h_n|^{-2}$. Consequently, $F(\mu)$ can be approximated as follows:
\begin{align} \nonumber 
F(\mu) \rightarrow&   \ln\left(1+|h_n|^2\frac{E\mu+\mu D_m\left(e^{\frac{N}{D_m}}-1\right)|h_n|^{-2}}{D_m\mu+1}\right)\\\nonumber &-  \mu  N \\ 
\rightarrow &   \frac{N}{D_m} -\frac{N\left(e^{\frac{N}{D_m}}-1\right)|h_n|^{-2}}{E}. 
\end{align}
As suggested by the title of Section \ref{subsectionxxx}, it is assumed that $ E>E_1$, which means  
\begin{align} \label{ea2}
F(\mu)  >0,   
\end{align}
for $\mu\rightarrow \left(e^{\frac{N}{D_m}}-1\right)|h_n|^{-2} E^{-1} $. 

When $\mu\rightarrow \infty$, $F(\mu)$ can be approximated as follows:
\begin{align} \label{ea3}
F(\mu)  \rightarrow&   \ln\left(1+|h_n|^2\frac{E +  D_m\left(e^{\frac{N}{D_m}}-1\right)|h_n|^{-2}}{D_m}\right) \\   \nonumber-  &\mu \left(N- D_m\ln\left(1+ e^{-\frac{N}{D_m}} |h_n|^2 \frac{E  }{D_m }\right)\right)  \rightarrow -\infty,
\end{align}
where the last step is obtained since for H-NOMA $\left(N- D_m\ln\left(1+ e^{-\frac{N}{D_m}} |h_n|^2 \frac{E  }{D_m }\right)\right) >0$. 
Combining \eqref{ea1}, \eqref{ea2}, and \eqref{ea3}, the proof for the first step is complete. The unique root of $F(\mu)$ for $\mu>\left(e^{\frac{N}{D_m}}-1\right)|h_n|^{-2} E^{-1}$ is denoted by  $\mu^*$, where $\mu^*>\left(e^{\frac{N}{D_m}}-1\right)|h_n|^{-2} E^{-1}$. 

{\it \underline{Step 2}} is to show that $T_n^*\triangleq \frac{1}{\mu^*}$ is the optimal solution of problem \eqref{5}. This step can be proved by using contradiction. Assume that the optimal solution of problem \eqref{5} is     smaller than $T_n^*$, denoted by $\bar{T}_n^*$. Denote $\bar{\mu}^*=\frac{1}{\bar{T}_n^*}$ and hence  $\bar{\mu}^*>\mu^*$. Following steps similar to those in \cite{Dinkelbach},  we have $F(\bar{\mu}^*)=0$, i.e.,     $\bar{\mu}^*$ is also a   root of $F(\mu)$ and $\bar{\mu}^*>\mu^*>\left(e^{\frac{N}{D_m}}-1\right)|h_n|^{-2} E^{-1}$, which   contradicts    the uniqueness of the   root of $F(\mu)$.  Step 2 means that finding the optimal solution of problem \eqref{5} is equivalent to finding  the largest   root of $F(\mu)$. Therefore, the optimality of   Newton's method can be straightforwardly shown, and in the following, we only focus on the modified Dinkelbach's method. 

{\it \underline{Step 3}} is to show $F(\mu)$  is a strictly  decreasing function of $\mu$, for $\mu^*\leq \mu\leq \infty$, which can be proved by using the facts that   $\mu^*$ is the unique  root of $F(\mu)$ for $\mu>\left(e^{\frac{N}{D_m}}-1\right)|h_n|^{-2} E^{-1}$,  and $F(\mu)$ is concave.

{\it \underline{Step 4}} is to show $F(\mu_{t+1})<0$ if $F(\mu_{t})<0$    for  the modified Dinkelbach's method\footnote{The simple proof derived for Lemma 5 in \cite{Dinkelbach} cannot be used  for the considered problem since the constraint set $\tilde{\mathcal{S}}(\mu)$ is now a function of $\mu$.  }. Since $F(\mu_{t})<0$ and $F(\mu)$ is a strictly decreasing function for $\mu\geq \mu^*$, we have $\mu_t>\mu^*$.   The convergence  analysis for   Newton's method is provided first to facilitate that of Dinkelbach's method.   The quadratic convergence analysis for Newton's  method yields the following \cite{Gaussiandd}
\begin{align}
\mu^*-\mu_{t+1} = -\frac{F''(\mu_{\xi})}{2F'(\mu_t)} (\mu^*-\mu_t)^2,
\end{align}
where $\mu^*\leq \mu_{\xi} \leq \mu_{t}$.   Note that $F'(\mu_t)<0$ and $F''(\mu_{\xi})<0$. Therefore, we have
\begin{align}
\mu^*-\mu_{t+1} < 0,
\end{align}
which means $F(\mu_{t+1})< 0$ for Newton's method. 

For   Dinkelbach's method,    $F(\mu)$ is first expressed  as  $F(\mu)=A(\mu)-\mu B(\mu)$, which means that $\mu$ is updated as  follows:
\begin{align}\label{+1 dm}
\mu_{t+1}=\frac{A(\mu_t)}{B(\mu_t)} = \mu_t+\frac{F(\mu_t) }{B(\mu_t)}.
\end{align}
 
Recall  that Newton's method  updates  $\mu$ as follows:
\begin{align}\label{+1 nm}
\mu_{t+1} = \mu_t -\frac{F(\mu_t)}{F'(\mu_t)}=\mu_t +\frac{F(\mu_t) }{B(\mu_t)-(A'(\mu_t)-\mu_tB'(\mu_t))}.
\end{align}
From \eqref{first derival}, $A'(\mu_t)-\mu_tB'(\mu_t)$ can be expressed as follows: 
\begin{align}
A'(\mu_t)-\mu_tB'(\mu_t)=\frac{|h_n|^2E+D_m\left(e^{\frac{N}{D_m}}-1\right)}{e^{\frac{N}{D_m}}D_m\mu +|h_n|^2E\mu+1}>0.
\end{align}
Note that $B(\mu)>(A'(\mu_t)-\mu_tB'(\mu_t))>0$ since $F'(\mu_t)<0$. 
Therefore, the step size of Newton's method is larger than that of Dinkelbach's method. In other words, starting with the same $\mu_t$,   $\mu_{t+1}$ obtained  from Newton's method is smaller than that of Dinkelbach's method, which means $F(\mu_{t+1})<0$ also holds for Dinkelbach's method.

{\it \underline{Step 5}} is to show $ \mu_{t+1}<\mu_{t}$, which can be proved  by using   \eqref{+1 dm} and \eqref{+1 nm}. An interesting  property by Steps 4 and 5 is that if the initial value of $\mu$ is set as $\infty$, $\mu_t$ is decreasing and approaches  $\mu^*$ as the number of  iterations increases, but will never pass $\mu^*$, as $F(\mu_t)$ is always negative. 

{\it \underline{Step 6}} is to show that Dinkelbach's method  converges to $\mu^*$.  This can be proved by contradiction. Assume  that the iterative algorithm converges to $\breve{\mu}$, i.e., $\underset{t\rightarrow \infty}{\lim}\mu_t=\breve{\mu}$ and $\breve{\mu}\neq \mu^*$. Although $F(\mu)$ might have more than one root, $F(\mu_{t})$ is always negative if $\mu_0=\infty$, as illustrated by Step 4. Therefore, $\mu_t$ is always larger than $\mu^*$, i.e.,   $\breve{\mu}>\mu^*$. Since $F(\mu)$ is strictly decreasing for $\mu>\mu^*$, we have the conclusion that $0=F(\breve{\mu})<F(\mu^*)=0$, which cannot be true. The proof is complete. 
\end{proof}

\begin{corollary}\label{corollary1}
The algorithm based on Newton's method converges faster than the one based on  Dinkelbach's method. 
\end{corollary} 
\begin{proof}
The corollary can be proved by using   Step 4 in the proof for Theorem \ref{theorem1}. 
\end{proof}

{\it Remark 2:} Following Step 4 in the proof for Theorem \ref{theorem1}, one can also establish the equivalence between     Dinkelbach's method and    Newton's method for the classical  problem  in \cite{Dinkelbach}. 
 
{\it Remark 3:} The rationale behind  the  existence condition of the H-NOMA solution,  i.e.,  $ E>E_1$,  can be explained as  follows. From the proof of   Theorem \ref{theorem1}, we learn that $F(\mu)$ is a decreasing function of $\mu$ for $\mu>\mu^*$. On the other hand, $P_{n,1}^{*}(\mu)$ and $P_{n,2}^{*}(\mu)$ are   increasing functions of $\mu$. So it is important to ensure that  when $\mu$ is reduced to a value $\bar{\mu}$, such that  $P_{n,1}^{*}(\bar{\mu})=0$, i.e.,   $\bar{\mu}= \frac{e^{\frac{N}{D_m}}-1}{E|h_n|^2}$, $F(\mu)$ needs to be positive. Otherwise, a positive root of $F(\mu)$ does not exist, and there is no feasible H-NOMA solution. By substituting $\bar{\mu}$ into \eqref{solution hnoma} and solving the inequality $F(\bar{\mu})>0$, the existence condition is obtained.   Similarly, one can find that $E\geq  N |h_n|^{-2}$ is the condition under which OMA is feasible. Hence, OMA-MEC is the only feasible option to minimize the delay if $  N |h_n|^{-2}\leq E\leq E_1$, i.e., there is very limited energy available for MEC offloading. 

 {\it Remark 4:}  For   the case $ E_1<E< E_2$, both OMA and NOMA are applicable.  Simulation results show that H-NOMA always yields less delay than OMA,  although we have yet to obtain a formal proof for this conjecture. 
 
 \section{Numerical Studies}
In this section, the performance of the proposed MEC offloading scheme is studied and compared by using computer simulations, where the normalized channel gains are adopted for the  purpose of clearly demonstrating  the impact of the channel conditions on the delay. In Fig. \ref{fig1}, the impact of NOMA on the MEC offloading delay is shown as a function of the   energy consumption. Note that the curves for NOMA-MEC are generated based on the combination of H-NOMA and pure NOMA, i.e., if  $E\geq D_m\left(e^{\frac{N}{D_m}}-1\right)e^{\frac{N}{D_m}} |h_n|^{-2}$, pure NOMA is used, otherwise H-NOMA is used. For a given amount of energy consumed, Fig. \ref{fig1} shows that the use of NOMA reduces the delay significantly. Particularly when there is plenty of energy available at user $n$, i.e., $E\geq D_m\left(e^{\frac{N}{D_m}}-1\right)e^{\frac{N}{D_m}} |h_n|^{-2}$, the use of NOMA ensures that $D_m$ is sufficient for offloading and there is no need to exploit extra time.  On the other hand, when the available energy  for MEC offloading is reduced, the delay performances of NOMA and OMA become similar.  Fig. \ref{fig2} provides a comparison of the convergence rates of the two proposed iterative algorithms. Note that both algorithms start with a delay of $0$ ($\mu_0=\infty$) and only the delay after convergence in the figure is achievable. As shown in the figure,  Dinkelbach's method converges generally  slowly than  Newton's method, as predicted by  Corollary \ref{corollary1}, although they perform the same  in conventional scenarios \cite{Dinkelbach}. 
 
\begin{figure}[!htbp]\centering \vspace{-1em}
    \epsfig{file=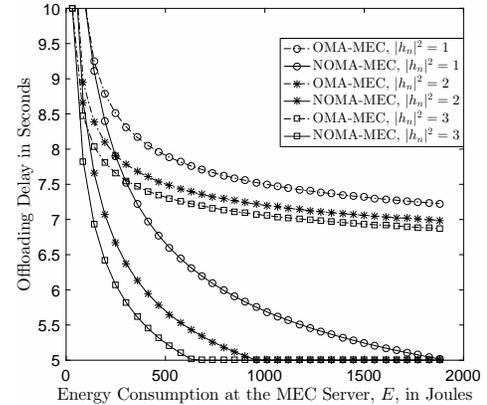, width=0.34\textwidth, clip=}\vspace{-1em}
\caption{ The impact of NOMA on MEC offloading delay. $N=15$ and $D_m=5$.    \vspace{-1em} }\label{fig1}
\end{figure}
 \vspace{-0.5em} 
 
\begin{figure}[!htbp]\centering \vspace{-1em}
    \epsfig{file=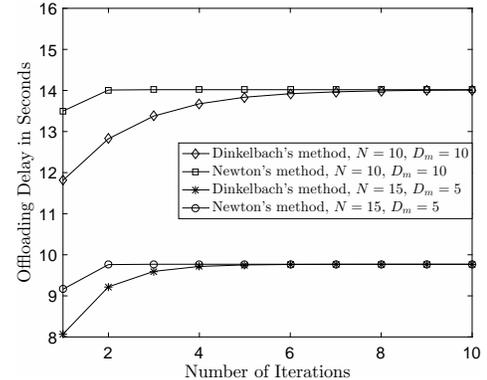, width=0.34\textwidth, clip=}\vspace{-1em}
\caption{ The convergence comparison between the two proposed iterative methods with $|h_n|^2=1$.     \vspace{-1em} }\label{fig2}
\end{figure}
\vspace{-1em}
\section{Conclusions}
In this paper, two iterative algorithms have been developed to minimize the offloading delay of NOMA-MEC. The optimality of the algorithms has been proven  and their rates of convergence were also analyzed.   Furthermore,   criteria for  choosing  between the three possible modes, OMA, pure NOMA, and  H-NOMA, for MEC offloading have been   established. 
    \bibliographystyle{IEEEtran}
\bibliography{IEEEfull,trasfer}

  \end{document}